\documentclass[12pt]{article}
\usepackage{amssymb,amsfonts,latexsym,amsmath,theorem,mathrsfs,graphicx,xcolor}
\usepackage[square,sort&compress,numbers]{natbib}
\usepackage{subcaption}
\usepackage{hyperref}
\usepackage[utf8]{inputenc}
\usepackage[english]{babel}

\newcommand{\braket}[1]{\langle#1\rangle}

\DeclareMathOperator{\SO}{SO}
\DeclareMathOperator{\SU}{SU}
\DeclareMathOperator{\U}{U}

\DeclareMathOperator{\Tr}{Tr}

\DeclareMathOperator{\Id}{Id}

\theoremstyle{plain}
\newtheorem{thm}{Theorem}

\newtheorem{proposition}[thm]{Proposition}

{\theorembodyfont{\rmfamily}

}

\newenvironment{proof}{\noindent{\it Proof:\, }}{\hfill$\Box$\vspace*{0.5cm}
}

\begin{document}

\title{\textbf{Skyrmion crystals stabilized by $\omega$-mesons}}

\author{\Large Derek Harland\footnote{d.g.harland@leeds.ac.uk},
Paul Leask\footnote{paulnleask@live.co.uk  (corresponding  author)},\: and Martin Speight\footnote{j.m.speight@leeds.ac.uk} \\ \\
School of Mathematics, University of Leeds, Leeds LS2 9JT, UK}
\date{\today}

\author{\Large Derek Harland\textsuperscript{1}\footnote{d.g.harland@leeds.ac.uk},
Paul Leask\textsuperscript{1,2}\footnote{palea@kth.se  (corresponding  author)},\: and Martin Speight\textsuperscript{1}\footnote{j.m.speight@leeds.ac.uk} \\ \\
\small \textsuperscript{1}\textit{School of Mathematics, University of Leeds, Leeds LS2 9JT, UK} \\
\small \textsuperscript{2}\textit{Department of Physics, KTH Royal Institute of Technology, 10691 Stockholm, Sweden}}
\date{\today}

\maketitle

\begin{abstract}
We investigate the ground state crystalline structure of nuclear matter in the $\omega$-meson variant of the Skyrme model.
After minimizing energy with respect to variations of both the Skyrme field and the period lattice, we find four distinct periodic solutions which are similar to those found in the standard Skyrme model.
We use these crystals to calculate coefficients in the Bethe--Weizs\"acker semi-empirical mass formula and the compression modulus of infinite nuclear matter, and find a significant improvement as compared with other variants of the Skyrme model.
\end{abstract}


\section{Introduction}

Skyrme models are a class of chiral Lagrangians in which baryons are modeled using topological solitons.
In common with all chiral Lagrangians, they can be considered low-energy descriptions of quantum chromodynamics (QCD).
The identification of solitons with baryons is justified by Witten's observation that the baryon number $B\in\mathbb{Z}$ can be identified with the degree of the chiral field $\varphi: \mathbb{R}^3 \rightarrow \SU(2)$ \cite{Witten_1983_1,Witten_1983_2}.

The simplest chiral Lagrangian, consisting of a nonlinear sigma model and a pion mass term, does not support stable solitons, so additional terms must be included in the Lagrangian in order to stabilize them.
Skyrme's proposal was the inclusion of a higher fourth-order term with opposing scaling behavior to provide the soliton with a scale \cite{Skyrme_1961}.
This is known as the standard Skyrme model, and it is widely studied as a model of atomic nuclei and dense nuclear matter.

Dense nuclear matter can be modeled in the standard Skyrme model using crystals \cite{Huidobro_2023}.
These are periodic solutions of the Euler--Lagrange equations that minimize the energy per unit cell.
In the model with massless pions, the crystal with the lowest energy per baryon number was discovered independently by Kugler and Shtrikman \cite{Kugler_1988} and Castillejo \textit{et al.}\ \cite{Castillejo_1989}.
This crystal resembles a cubic lattice, with each vertex carrying baryon number $1/2$, and is referred to as a lattice of half-skyrmions, or the $\textup{SC}_{1/2}$ crystal. 
Crystals in the model with non-zero pion mass have been investigated recently \cite{Leask_2023}.
The pion mass breaks chiral symmetry, and as a result the $\textup{SC}_{1/2}$ crystal degenerates to four distinct crystals with slightly different energies.
Two of these crystals, including the one with the lowest energy, do not enjoy cubic lattice symmetry and their fundamental domains are cuboidal but not cubic.
The discovery of these non-cubic crystals was enabled by a new energy-minimization algorithm that allows the lattice structure, as well as the Skyrme field, to vary.

This paper concerns a variant of the Skyrme model that was first proposed by Adkins and Nappi \cite{Nappi_1984}.
This model does not include a Skyrme term; instead, solitons are stabilized by an $\omega$-meson field that is coupled anomalously to the chiral field $\varphi$ through the Wess--Zumino term.
While it is well-motivated, this model initially received less attention than the standard Skyrme model because it has proved much harder to find soliton solutions.
The technical reason for this is that the energy is not bounded from below, rendering gradient descent-based energy minimization algorithms useless.
Interest in the model was revived in \cite{Sutcliffe_2009}, which constructed the first topological solitons with $B>1$, albeit within the rational map approximation.
Progress has also been made in the $B=1$ sector \cite{Speight_2018}, where it was shown that a simple perturbation of the model can reproduce the neutron-proton mass difference.

In a recent paper \cite{Gudnason_2020} a new method for constructing solitons in this model was developed.
Therein, they found true static solutions for topological charges $1$ through $8$, for a range of coupling constants.
Suitably calibrated, the model reproduces several properties of atomic nuclei with reasonable accuracy.
In particular, its classical binding energies are comparable with experimental values, which is not the case in the standard Skyrme model.
With this calibration the shapes of solitons are sometimes different from the standard Skyrme model, and from the predictions of the rational map approximation.

This paper presents an investigation of crystals in the $\omega$-meson Skyrme model, using the methods of \cite{Leask_2023} and \cite{Gudnason_2020}.
As in the standard Skyrme model, we find four distinct crystals, and the crystal with the lowest energy differs from  the crystal of \cite{Kugler_1988,Castillejo_1989} both in its iso-orientation and its symmetries.
From these crystals we are able to calculate coefficients in the Bethe--Weizs\"acker semi-empirical mass formula and the nuclear matter incompressibility coefficient.
In both cases we obtain more acceptable values than had been obtained in other Skyrme models.
Furthermore, the method developed herein has also been utilized to determine crystals in the baby Skyrme model coupled to the $\omega$-meson \cite{leask2024baby}.

Skyrmion crystals involving vector mesons have been studied elsewhere in the literature, for example in \cite{Park_2004,Harada_2013}.
These papers did not use the Adkins--Nappi model considered here but instead investigated more complicated Skyrme models: the model in \cite{Park_2004} included $\rho$-mesons and scalar fields in addition to the pion and $\omega$ fields, while the model in \cite{Harada_2013} was based on the entirely different framework of hidden local symmetry (HLS) and holography.
In the HLS approach, the hidden symmetry of the nonlinear $\sigma$ model is gauged and the corresponding gauge particle acquires mass through the Higgs mechanism \cite{Forkel_1991}.
This allows for the incorporation of $\rho$-mesons, as well as the $\omega$-meson.
Nevertheless, in both of these papers, the Kugler--Shtrikman Fourier series method \cite{Kugler_1988} was generalized to incorporate vector mesons.
This method assumed a cubic lattice symmetry from the outset and did not allow the lattice geometry to vary, so it is was not able to find the new lower-energy crystals discovered in this paper.
The recent paper \cite{Barriga:2023jam} constructs exact solutions to the Adkins--Nappi model on compact domains.
However, unlike the crystals studied here, these solutions are time-dependent.
When extended periodically, the solutions of \cite{Barriga:2023jam} have zero baryon number per unit period, whereas our skyrmion crystals have non-zero baryon number per unit period.

The $\omega$-meson model is presented in the next section, where we also present a new topological energy bound valid for crystals and, more generally, for skyrmions on compact domains.
In section \ref{sec:Stress-energy tensor} we present our energy-minimization algorithm and compute the relevant stress-energy tensor.
We present our new crystal solutions in sections \ref{sec:Results}, and proceed to calculate the semi-empirical mass formula and nuclear matter incompressibility coefficient in sections \ref{sec:SEMF} and \ref{sec:incompressibility}.
We draw conclusions in section \ref{sec:conclusion}.


\section{The $\omega$-Skyrme model}

The $\omega$-meson variant of the Skyrme model is a non-linear sigma model coupled to the isoscalar $\omega$ vector meson field.
It consists of the Skyrme field $\phi: \mathbb{R}^{1,3} \rightarrow \SU(2)$ and the $\omega$ vector meson, which is a 1-form on $\mathbb{R}^{1,3}$.
Here $\mathbb{R}^{1,3}=\mathbb{R}\times\mathbb{R}^3$ is Minkowski space with metric $\eta$ and metric signature $+---$.
The $\omega$-Skyrme Lagrangian defined by Adkins \& Nappi \cite{Nappi_1984} is given by
\begin{equation}
\label{eq: Omega Lagrangian}
    \mathcal{L} = \mathcal{L}_{\phi} + \mathcal{L}_{\omega} + \mathcal{L}_{\textup{WZ}}.
\end{equation}
Here $\mathcal{L}_{\phi}$ is the sigma model Lagrangian with the explicit chiral symmetry breaking pion mass term,
\begin{equation}
    \mathcal{L}_{\phi} = -\frac{F_{\pi}^2 m_{\pi}^2}{8\hbar^3} \Tr\left( \Id_2 - \phi\right) - \frac{F_{\pi}^2}{16\hbar}\eta^{\mu\nu} \Tr(L_\mu L_\nu), \quad L_\mu=\phi^\dagger\partial_\mu\phi.
\end{equation}
The minimally broken $\U(1)_V$ Lagrangian for spin-1 mesons is given by the term
\begin{equation}
\label{eq: Omega Lagrangian term}
    \mathcal{L}_{\omega} = \frac{m_{\omega}^2}{2\hbar^3}\eta^{\mu\nu} \omega_\mu \omega_\nu - \frac{1}{4\hbar} \eta^{\mu\alpha}\eta^{\nu\beta} \omega_{\mu\nu} \omega_{\alpha\beta},\quad \omega_{\mu\nu}=\partial_\mu\omega_\nu-\partial_\nu\omega_\mu,
\end{equation}
and the gauged Wess-Zumino term is
\begin{equation}
\label{eq: Wess-Zumino Lagrangian term}
    \mathcal{L}_{\textup{WZ}} = \beta_{\omega} \omega_\mu \mathcal{B}^\mu, \quad \mathcal{B}^\mu = \frac{1}{24 \pi^2\sqrt{-\eta}} \epsilon^{\mu \nu \rho \sigma} \Tr(L_\nu L_\rho L_\sigma),
\end{equation}
which describes the the coupling of the $\omega$-meson to three pions.
The baryon number can be identified with a topological charge and is given by
\begin{equation}
\label{eq: Baryon number}
    B = \int_{\mathbb{R}^3} \textup{d}^3x \sqrt{-\eta} \, \mathcal{B}^0.
\end{equation}

The main free parameters of this model are the pion decay constant $F_{\pi}$, the pion mass $m_{\pi}$, the $\omega$-meson mass $m_{\omega}$, and the coupling constant $\beta_{\omega}$, while $\hbar=197.3$ MeV fm is the reduced Planck constant.
The coupling constant $\beta_{\omega}$ can be related to the $\omega \rightarrow \pi^+\pi^-\pi^0$ decay rate, which is in reality enhanced by the resonance $\omega \rightarrow \rho + \pi$, but is not included in the current theory.
The decay rate, calculated using fiducial experimental values, is found to be $\Gamma_{\omega \rightarrow 3\pi}=8.49$ MeV, which gives the upper bound $\beta_{\omega} \leq 23.9$ \cite{Gudnason_2020}.

For convenience, we follow Sutcliffe \cite{Sutcliffe_2009} and rescale the $\omega$ meson by $\omega \mapsto \omega F_{\pi}$, and choose the classical energy scale to be $\tilde{E}=F_{\pi}^2/m_{\omega}$ (MeV) and the length scale to be $\tilde{L}=\hbar/m_{\omega}$ (fm).
Then the rescaled $\omega$-Skyrme Lagrangian in dimensionless units is given by
\begin{align}
\label{eq:Lagrangian density}
    \mathcal{L} = \, & -\frac{m^2}{8}  \Tr\left( \Id_2 - \phi \right) -\frac{1}{16}\eta^{\mu\nu} \Tr(L_\mu L_\nu) + \frac{1}{2} \eta^{\mu\nu} \omega_\mu \omega_\nu - \frac{1}{4} \eta^{\mu\alpha}\eta^{\nu\beta} \omega_{\mu\nu} \omega_{\alpha\beta} + c_{\omega} \omega_\mu \mathcal{B}^\mu,
\end{align}
where the rescaled pion mass and $\omega$ coupling constant are, respectively, $m = m_{\pi}/m_{\omega}$ and $c_{\omega}=m_{\omega}\beta_{\omega}/F_{\pi}$.
The energy-momentum tensor (in dimensionless Skyrme units) is given by
\begin{align}
    T_{\mu\nu} = \, & \frac{2}{\sqrt{-\eta}}\frac{\partial (\sqrt{-\eta}\mathcal{L})}{\partial \eta^{\mu\nu}} = 2\frac{\partial \mathcal{L}}{\partial \eta^{\mu\nu}} - \eta_{\mu\nu}\mathcal{L} \nonumber \\
    = \, & -\frac{1}8 \Tr(L_\mu L_\nu) + \omega_\mu \omega_\nu - \eta^{\alpha\beta} \omega_{\mu\alpha} \omega_{\nu\beta} - \eta_{\mu\nu} \left\{-\frac{m^2}{8} \Tr\left( \Id_2 - \phi \right) \right. \nonumber \\
    \, & \left.- \frac{1}{16} \eta^{\alpha\beta} \Tr(L_\alpha L_\beta) + \frac{1}{2}\eta^{\alpha\beta} \omega_\alpha \omega_\beta - \frac{1}{4} \eta^{\alpha\rho} \eta^{\beta\sigma} \omega_{\alpha\beta}\omega_{\rho\sigma} \right\}.
\end{align}
Notice that the Wess-Zumino Lagrangian makes no contribution to this because it does not depend on the metric tensor.
The energy functional is obtained from the temporal part of the energy-momentum tensor.  For the Minkowski metric this is
\begin{align}\label{eq:energy density}
    \mathcal{E} = \frac{m^2}{8} \Tr\left( \Id_2 - \phi \right) - \frac{1}{16} \Tr(L_i L_i+L_0L_0) + \frac{1}{2} \omega_0^2 + \frac{1}{2} \omega_i\omega_i + \frac{1}{2} \omega_{0i} \omega_{0i}+ \frac{1}{4} \omega_{ij} \omega_{ij}.
\end{align}

We are only interested in finding static solutions, so we write $\phi(x,t)=\varphi(x)$, where the map $\varphi: \mathbb{R}^3 \rightarrow \SU(2)$ will now be identified as the Skyrme field.
In particular, we will study this model on the physical space $\mathbb{R}^3$ under the assumption of periodicity with respect to some $3$-dimensional lattice
\begin{equation}
    \Lambda = \left\{ n_1 \vec{X}_1 + n_2 \vec{X}_2 + n_3 \vec{X}_3: n_i \in \mathbb{Z} \right\}.
\end{equation}
We do so by interpreting the domain of the fields $\varphi,\omega$ as $\mathbb{R}^3/\Lambda$, where $(\mathbb{R}^3/\Lambda,d)$ is a $3$-torus equipped with the standard Euclidean metric $d$.  It will prove convenient to identify this domain with the unit $3$-torus by $\mathbb{T}^3\equiv S^1 \times S^1 \times S^1 = \mathbb{R}^3/\mathbb{Z}^3$ via the obvious diffeomorphism
\begin{equation}
    F: \mathbb{T}^3 \rightarrow \mathbb{R}^3/\Lambda, \quad (x_1,x_2,x_3) \mapsto x_1\vec{X}_1 + x_2\vec{X}_2 + x_3\vec{X}_3.
\end{equation}
The Euclidean metric $d$ on $\mathbb{R}^3/\Lambda$ can be identified with the pullback metric $g$ on $\mathbb{T}^3$, i.e.\
\begin{equation}
    g = F^*d = g_{ij} \textup{d}x_i \textup{d}x_j, \quad g_{ij} = \vec{X}_i \cdot \vec{X}_j.
\end{equation}
Varying the lattice $\Lambda$ is then equivalent to varying the flat metric $g$ on $T^3$.

We will write the Skyrme field using pion field notation, that is, we write $\varphi = \varphi_0 \Id_2 + i \varphi_j \tau^j$ where $\tau^j$ are the usual Pauli spin matrices.
Then, we identify $\SU(2)$ with $S^3$ via the isometry
\begin{equation}
    \SU(2) \ni \begin{pmatrix}
        \varphi_0 + i\varphi_3 & i\varphi_1 + \varphi_2 \\
        i\varphi_1 - \varphi_2 & \varphi_0 - i\varphi_3
    \end{pmatrix} \leftrightarrow (\varphi_0,\varphi_1,\varphi_2,\varphi_3) \in S^3,
\label{eq: Quaternion representation}
\end{equation}
The three fields $\varphi_1,\varphi_2,\varphi_3$ are identified with pions, and the field $\varphi_0$ is sometimes referred to as the $\sigma$-field and is constrained by the equation $\varphi_A\varphi_A=1$, where the repeated index is summed over $A=0,1,2,3$.

Since we are only interested in static field configurations, only the temporal component of the topological current remains, i.e. $\mathcal{B}^i=0$.
Consequently, only the temporal component $\omega_0$ of the $\omega$-meson survives, since the topological charge density acts as a source term for the $\omega$ field.
For notational convenience, we will drop the subscript and denote $\omega \equiv \omega_0$.
With these conventions, the static Lagrangian and energy functionals obtained by integrating \eqref{eq:Lagrangian density} and \eqref{eq:energy density} over one period are
\begin{align}
\label{eq:unbounded energy}
    -L(\varphi, \omega,g) &= \int_{\mathbb{T}^3} \textup{d}^3x \sqrt{g} \, \left\{ \frac{1}{4} m^2 (1 - \varphi_0) + \frac{1}{8} g^{ij} \partial_i \varphi_A \partial_j \varphi_A - \frac{1}{2} g^{ij} \partial_i \omega \partial_j \omega - \frac{1}{2} \omega^2 - c_{\omega} \omega \mathcal{B}_0 \right\},\\
\label{eq: Static energy}
    E(\varphi, \omega,g) &= \int_{\mathbb{T}^3} \textup{d}^3x \sqrt{g} \, \left\{ \frac{1}{4} m^2 (1 - \varphi_0) + \frac{1}{8} g^{ij} \partial_i \varphi_A \partial_j \varphi_A + \frac{1}{2} g^{ij} \partial_i \omega \partial_j \omega + \frac{1}{2} \omega^2 \right\}.
\end{align}
Solutions of the Euler-Lagrange equations are critical points of the Lagrangian \eqref{eq:unbounded energy}.
Since this is not bounded from above or below it is not amenable to standard energy-minimisation methods.
Following \cite{Gudnason_2020} we reformulate it using the Euler-Lagrange equation corresponding to temporal $\omega$,
\begin{equation}
\label{eq: Omega field equation}
    \left(-g^{ij} \partial_i \partial_j + 1\right) \omega = -c_{\omega}\mathcal{B}_0.
\end{equation}
This is a linear equation for $\omega$ with a source term proportional to the baryon current.
The $\omega$-meson is completely determined by the Skyrme field $\varphi$ and the domain metric $g$.
Taking the inner product of \eqref{eq: Omega field equation} with $\omega$ and integrating by parts yields
\begin{align}
\label{eq: IBP result}
    \int_{\mathbb{T}^3} \textup{d}^3x \sqrt{g} \, c_{\omega}\omega\mathcal{B}_0 = -\int_{\mathbb{T}^3} \textup{d}^3x \sqrt{g} \, \left( g^{ij} \partial_i \omega \partial_j \omega + \omega^2 \right).
\end{align}
It follows that $-L$ given \eqref{eq:unbounded energy} is equal to the energy \eqref{eq: Static energy} when $\omega$ satisfies the constraint \eqref{eq: Omega field equation}.
Extremising the unbounded functional \eqref{eq:unbounded energy} with respect to variations $\varphi$ and $\omega$ is equivalent to extremising the bounded energy \eqref{eq: Static energy} subject to the constraint \eqref{eq: Omega field equation}.

We wish to find skyrmion crystals, i.e.\ static periodic solutions $\varphi,\omega$ of the Euler-Lagrange equations whose energy is minimized with respect to variations of the period lattice $\Lambda$.  We do so by minimizing \eqref{eq: Static energy} with respect to variations in $\varphi,g$, with $\omega$ being determined by the constraint \eqref{eq: Omega field equation}.  We will describe a numerical method for doing so in the next section.



Before moving on, it is interesting to note that the energy \eqref{eq: Static energy} subject the constraint \eqref{eq: Omega field equation} obeys a topological energy bound.  In fact, the bound is valid in the more general setting of maps $\varphi: M \rightarrow N$ between compact Riemannian $3$-manifolds $(M^3,g)$ and $(N^3,h)$, so we reformulate the energy in this more general setting.  We have a functional given by
\begin{equation}
\label{eq: Main energy functional}
E(\varphi,g) = \int_{M} \left( \frac{1}{8} \left|\textup{d}\varphi\right|^2_g + \frac{1}{4}(V\circ\varphi) 
+ \frac{1}{2}\left|\textup{d}\omega\right|^2_g + \frac{1}{2}\omega^2 \right) \textup{vol}_g,
\end{equation}
subject to the constraint
\begin{equation}
\label{eq: General constraint}
    \left( \Delta_g + 1 \right) \omega = - c_{\omega} \ast \varphi^* \Omega,
\end{equation}
where $\Omega$ is the normalized volume form on $N$, i.e.\
\begin{equation}
    \Omega = \frac{\textup{vol}_h}{|N|}.
\end{equation}

\begin{proposition}\label{prop:bound}
The energy \eqref{eq: Main energy functional} subject to the constraint \eqref{eq: General constraint} satisfies the topological energy bound,
\begin{equation}
    E \geq \frac{B^2 c_{\omega}^2}{2|M|},
\end{equation}
in which $B$ is the topological charge (i.e.\ degree) of $\varphi:M\to N$ and $|M|$ is the volume of $M$.
\end{proposition}
\begin{proof}
Let us define $\mathcal{B}=\ast \varphi^* \Omega$ such that $\varphi^* \Omega = \mathcal{B}\, \textup{vol}_{g}$.
Then, from the $\omega$-meson constraint \eqref{eq: General constraint}, the topological charge can be expressed as
\begin{equation}
    B = \int_{M} \varphi^* \Omega = -\frac{1}{c_{\omega}} \int_{M} \left( \Delta_g + 1 \right) \omega \, \textup{vol}_g = -\frac{1}{c_{\omega}} \int_{M} \omega \, \textup{vol}_g.
\end{equation}
Using the Cauchy--Schwartz inequality, we obtain the following relation
\begin{equation}
    B^2 = \frac{1}{c_{\omega}^2} \left( \int_{M} \omega \, \textup{vol}_g \right)^2 \leq \frac{1}{c_{\omega}^2} \left( \int_{M} \omega^2 \, \textup{vol}_g \right) \left( \int_{M} 1 \, \textup{vol}_g \right) = \frac{|M|}{c_{\omega}^2} \int_{M} \omega^2 \, \textup{vol}_g.
\end{equation}
With this, we can derive a simple lower topological bound on the static energy \eqref{eq: Main energy functional}, that is,
\begin{equation}
    E \geq \frac{1}{2} \int_{M} \omega^2 \, \textup{vol}_g \geq \frac{B^2 c_{\omega}^2}{2|M|}
\end{equation}
\end{proof}

For the particular case of interest, $M=\mathbb{T}^3$ with flat metric given by a matrix $g$, the bound is
\begin{equation}\label{eq:bound}
    E \geq E_{\textup{bound}} = \frac{B^2 c_{\omega}^2}{2\sqrt{g}}.
\end{equation}


\section{Energy minimization and the stress-energy tensor}
\label{sec:Stress-energy tensor}

We now turn to the problem of constructing skyrmion crystals, i.e.\ minimizing the energy \eqref{eq: Static energy} with respect to variations in $\varphi$ and $g$.  We do this numerically, using arrested Newton flow.  This algorithm works by solving Newton's equations of motion for the energy $E$, written formally as:
\begin{equation}\label{eq: Newton flow}
    \frac{\textup{d}^2}{\textup{d}t^2}(\varphi_A,g_{ij}) = -\nabla E.
\end{equation}
Initial conditions are chosen such that $\frac{\textup{d}}{\textup{d}t}(\varphi_A,g_{ij})=0$.
These ensure that the flow reduces energy at early times.
If at any later time the energy begins to increase, the flow is arrested and the velocities $\frac{\textup{d}}{\textup{d}t}(\varphi_A,g_{ij})$ are set to zero.
The flow then resumes from the same position.
It is deemed to have converged when $\nabla E$ is sufficiently small.

We recall that $\omega$ appearing in the energy functional \eqref{eq: Static energy} depends on $\varphi$ and $g$ through the constraint \eqref{eq: Omega field equation}.
Thus computing $E$ and its gradient entails computing $\omega$ at each time step.
As in \cite{Gudnason_2020}, this is accomplished using a conjugate gradient method.
The constraint \eqref{eq: Omega field equation} means that the metric-dependence of the energy is much more complicated than in the standard Skyrme model.
As a result, the algorithm described here is slightly different from the algorithm used earlier to find crystals in the standard Skyrme model \cite{Leask_2023}.

The gradient on the right hand side of \eqref{eq: Newton flow} is understood using the calculus of variations.
We write $\nabla E=(\Phi_A,S_{ij})$, in which $\Phi_A$ and $S_{ij}$ are defined by
\begin{equation}
    \frac{\textup{d}}{\textup{d}s}E(\varphi_s,g_s)\bigg|_{s=0} = \int_{\mathbb{T}^3} \textup{d}^3x \sqrt{g} \,\left(  \Phi_A(\varphi,g)\dot{\varphi}_A + S_{ij}(\varphi,g)\dot{g}_{kl}g^{jk}g^{li}\right)
\end{equation}
for all one-parameter variations $\varphi_s,g_s$ with $(\varphi_0,g_0)=(\varphi,g)$ and $\frac{\textup{d}}{\textup{d}s}(\varphi_s,g_s)=(\dot{\varphi},\dot{g})$ at $s=0$.
The calculation of $\Phi_A$ and $S_{ij}$ is delicate, because $\omega$ appearing in \eqref{eq: Static energy} depends on $\varphi$ and $g$ implicitly through the constraint \eqref{eq: Omega field equation}.
Using results of \cite{Gudnason_2020}, $\Phi_A$ is given (in the case of flat metrics on $\mathbb{T}^3$) by
\begin{equation}
\label{eq: Euler-Lagrange equations}
    \Phi_A = - \frac{1}{4}(\delta_{AB}-\varphi_A\varphi_B) (m^2 \delta_{0B} + g^{ij} \partial_i\partial_j \varphi_B) + \frac{c_{\omega}}{4\pi^2 \sqrt{g}} \epsilon_{ijk}\epsilon_{ABCD} \varphi_B \partial_i \omega \partial_j \varphi_C \partial_k \varphi_D,
\end{equation}
where $A,B,C,D=0,1,2,3$.
This coincides with the Euler--Lagrange equation of the original unconstrained energy functional \eqref{eq:unbounded energy}.
The stress-energy tensor $S_{ij}$ is computed in the following proposition, formulated in the general setting of maps $\varphi:(M,g)\to(N,h)$ between Riemannian 3-manifolds.

\begin{proposition}
The stress-energy tensor $S=S_{ij}\textup{d}x^i\textup{d}x^j$ associated to the energy \eqref{eq: Main energy functional} subject to the constraint \eqref{eq: General constraint} is the section of $\mathrm{Sym}^2(T^\ast M)$ given by
\begin{equation}
\label{eq: Stress-energy tensor}
    S(\varphi,g) = \left( \frac{1}{16}|\textup{d}\varphi|^2_g + \frac{1}{8}(V\circ\varphi) -  \frac{1}{4} |\textup{d}\omega|^2_g -\frac{1}{4} \omega^2 \right) g - \left(\frac{1}{8} \varphi^*h - \frac{1}{2} \textup{d}\omega\otimes\textup{d}\omega \right).
\end{equation}
\end{proposition}
Note that in local coordinates the formula \eqref{eq: Stress-energy tensor} gives
\begin{multline}
    S_{ij} = \left( \frac{1}{16}g^{mn} \partial_m\varphi_A\partial_n\varphi_A + \frac{1}{8}m^2\left(1-\varphi_0\right) - \frac{1}{4}g^{mn}\partial_m\omega\partial_n\omega - \frac{1}{4}\omega^2 \right) g_{ij}  \\
    - \frac{1}{8} \partial_i\varphi_A\partial_j\varphi_A + \frac{1}{2} \partial_i\omega\partial_j\omega.
\end{multline}
This coincides with the stress tensor for the original unconstrained energy functional \eqref{eq:unbounded energy}.
\begin{proof}
Let us introduce the notation $\langle A,B\rangle_g=A_{ij}B_{kl}g^{ik}g^{jl}$ for the natural inner product of two-tensors $A=A_{ij}\textup{d}x^i\textup{d}x^j$, $B=B_{kl}\textup{d}x^k\textup{d}x^l$.  The variation of the inverse metric and the volume form are given by
\begin{equation}
    \left.\frac{\textup{d}}{\textup{d}s}\right|_{s=0} g^{ij}(s) = - g^{ik}\dot{g}_{kl} g^{lj},\quad 
    \left.\frac{\textup{d}}{\textup{d}s}\right|_{s=0} \textup{vol}_{g_s} = \frac{1}{2}\langle g,\dot g\rangle_g \textup{vol}_g.
\end{equation}
These lead to the standard result for the first variation of the terms in \eqref{eq: Main energy functional} involving $\varphi$:
\begin{equation}
\label{eq: Phi variation}
\left.\frac{d}{ds}\int_{M} \left( \frac{1}{8} \left|\textup{d}\varphi\right|^2_g + \frac{1}{4}(V\circ\varphi) \right) \textup{vol}_g\right|_{s=0}=
\int_{M} \left\langle \frac{1}{16}|\textup{d}\varphi|^2_g g + \frac{1}{4}(V\circ\varphi)g - \frac{1}{8} \varphi^*h\,,\, \dot{g}\right\rangle_g \textup{vol}_g.
\end{equation}
It remains to compute the first variation of the terms in \eqref{eq: Main energy functional} involving $\omega$, which are more conveniently written using the constraint \eqref{eq: General constraint}:
\begin{equation}
    E^{\omega}(\omega,g) = \int_{M} \left( \frac{1}{2}\left|\textup{d}\omega\right|^2_g + \frac{1}{2}\omega^2 \right) \textup{vol}_g = - \frac{c_{\omega}}{2} \int_{M} \omega \varphi^*\Omega.
\end{equation}
Since the pullback $\varphi^*\Omega \in \Omega^3(M)$ is $g$-independent, the first variation of this with respect to the metric $g_s$ is given by
\begin{equation}
\label{eq: First omega variation}
    \left.\frac{\textup{d} E^{\omega}(\omega_s, g_s)}{\textup{d}s}\right|_{s=0}
    = -\frac{c_{\omega}}{2} \int_{M} \dot{\omega} \varphi^*\Omega
    = \frac{1}{2} \int_{M} \dot{\omega}\left( \Delta_g + 1 \right) \omega\,\textup{vol}_g 
    = \frac{1}{2} \int_{M} {\omega}\left( \Delta_g + 1 \right) \dot{\omega}\,\textup{vol}_g,
\end{equation}
where we have denoted $\dot{\omega}=\left.\frac{\textup{d}}{\textup{d}s}\right|_{s=0}\omega_s$.
This can be simplified as follows.
Consider the variation of the Hodge star operator $*_g: \Omega^3(M) \rightarrow \Omega^0(M)$,
\begin{equation}
\label{eq: Hodge star variation}
    \left.\frac{\textup{d}}{\textup{d}s}\right|_{s=0}*_{g_s} = -\frac{1}{2}\braket{g, \dot{g}}_g *_g,
\end{equation}
and define $\dot{\Delta}_g = \left.\frac{\textup{d}}{\textup{d}s}\right|_{s=0}\Delta_{g_s}$.
Then varying the $\omega$-meson constraint \eqref{eq: General constraint} and using \eqref{eq: Hodge star variation} yields
\begin{equation}
    \left( \Delta_g + 1 \right) \dot{\omega} = -\dot{\Delta}_g \omega + \frac{c_{\omega}}{2}\langle g,\dot g\rangle_g *_g \varphi^*\Omega
\end{equation}
Hence the first variation \eqref{eq: First omega variation} becomes
\begin{equation}
\label{eq: Omega variation reduced}
    \left.\frac{\textup{d} E^{\omega}(\omega_s, g_s)}{\textup{d}s}\right|_{s=0}
    = \frac{c_\omega}{4} \int_{M}\langle g,\dot g\rangle_g\omega\,\phi^\ast\Omega - \frac{1}{2} \int_{M} {\omega} \dot{\Delta}_g{\omega}\,\textup{vol}_g.
\end{equation}
To simplify the second term, we vary the identity,
\begin{equation}
    \int_{M} f ( \Delta_{g_s} f ) \textup{vol}_{g_s} = \int_{M} \braket{\textup{d}f,\textup{d}f}_{g_s} \textup{vol}_{g_s},
\end{equation}
to obtain
\begin{equation}
\label{eq: Laplacian variation relation}
\int_{M} \left\{ f \dot{\Delta}_g f + \frac{1}{2} f(\Delta_g f) \langle g,\dot{g}\rangle_g \right\} \textup{vol}_g 
    = \int_{M} \left\{ -\langle \textup{d}f\otimes \textup{d}f,\dot{g}\rangle_g + \frac{1}{2}|\textup{d}f|^2_g \langle g,\dot{g}\rangle_g \right\} \textup{vol}_g,
\end{equation}
valid for all smooth functions $f$.  Using particular case $f=\omega$ of \eqref{eq: Laplacian variation relation} and the constraint \eqref{eq: General constraint}, \eqref{eq: Omega variation reduced} rearranges to
\begin{equation}
    \left.\frac{\textup{d} E^{\omega}(\omega_s, g_s)}{\textup{d}s}\right|_{s=0} = \int_{M} \left\langle \frac{1}{2} \textup{d}\omega\otimes\textup{d}\omega - \frac{1}{4}|\textup{d}\omega|^2_g g - \frac{1}{4}\omega^2 g\,,\, \dot{g} \right\rangle_g \textup{vol}_g.
\end{equation}
Combining this with \eqref{eq: Phi variation}, the variation of $E$ takes the form $\int_M\langle S,\dot g\rangle_g\textup{vol}_g$, with $S$ given in \eqref{eq: Stress-energy tensor}.
%
\end{proof}

\section{Skyrmion crystals coupled to $\omega$-mesons}
\label{sec:Results}

\begin{figure}[t]
	\centering
	\begin{subfigure}[b]{0.45\textwidth}
	\includegraphics[width=\textwidth]{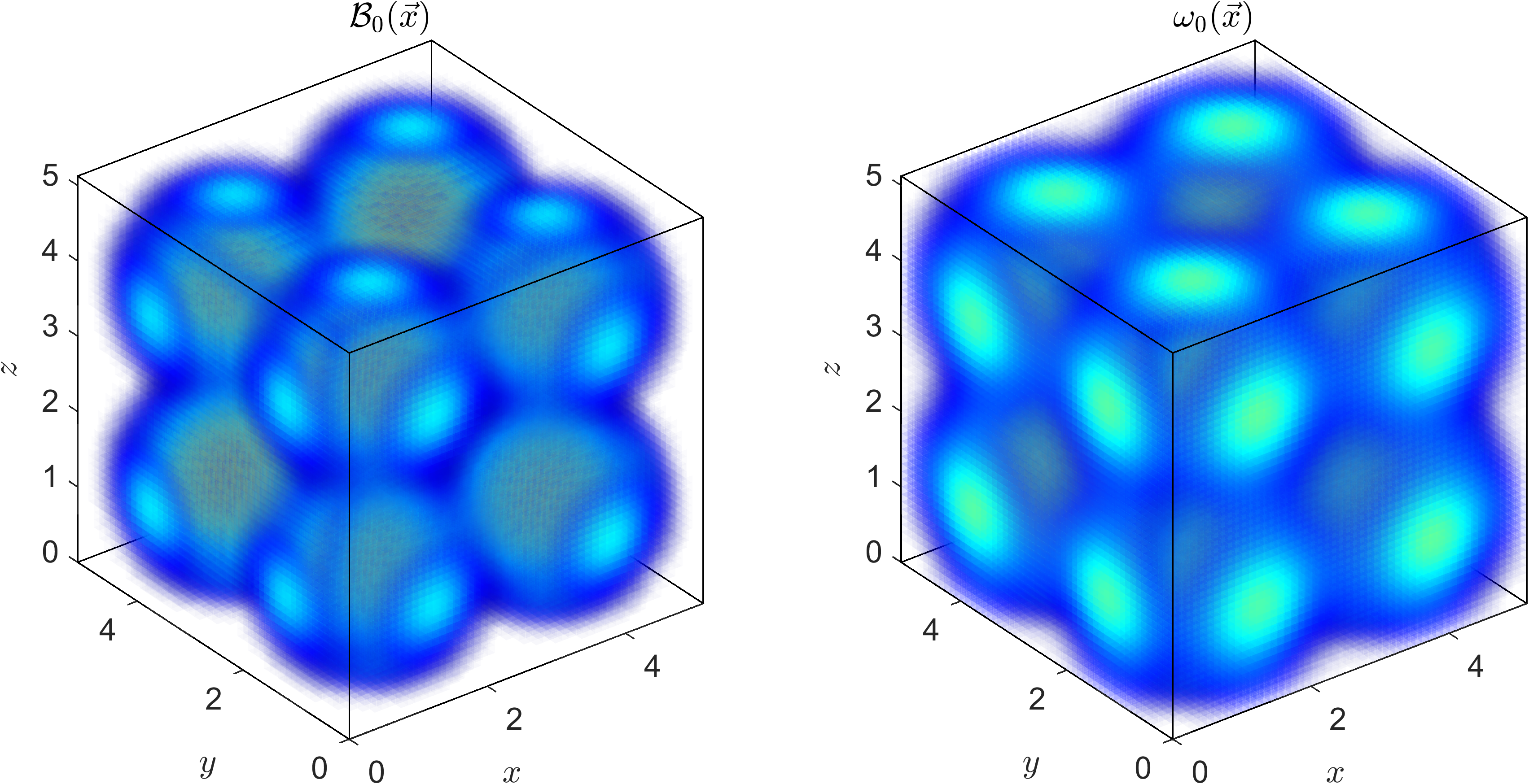}
	\caption{$\textup{SC}_{1/2}$ crystal}
	\label{fig: FCC}
	\end{subfigure}
	~
	\begin{subfigure}[b]{0.45\textwidth}
	\includegraphics[width=\textwidth]{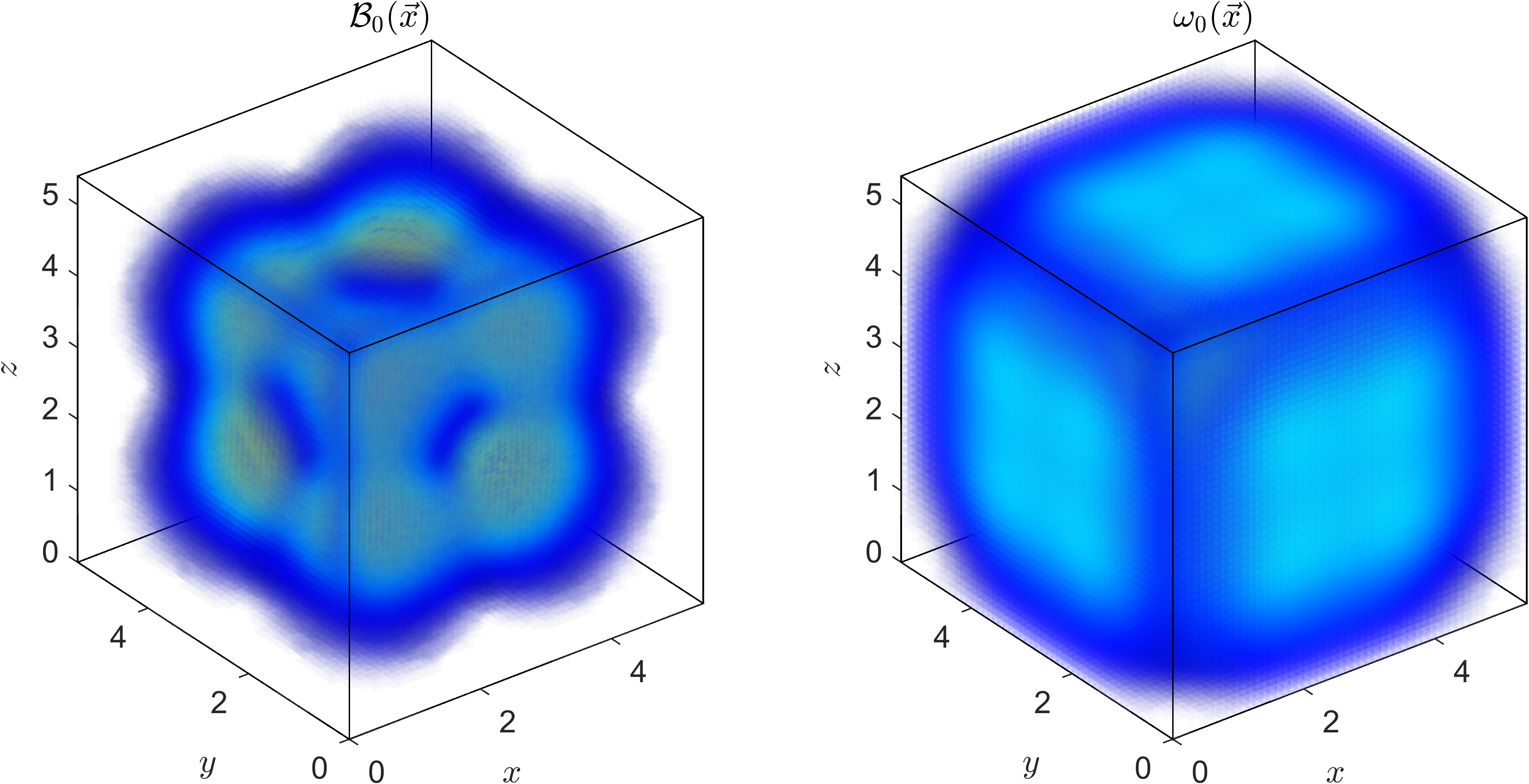}
	\caption{$\alpha$ crystal}
	\label{fig: Alpha}
	\end{subfigure} \\
	\begin{subfigure}[b]{0.45\textwidth}
	\includegraphics[width=\textwidth]{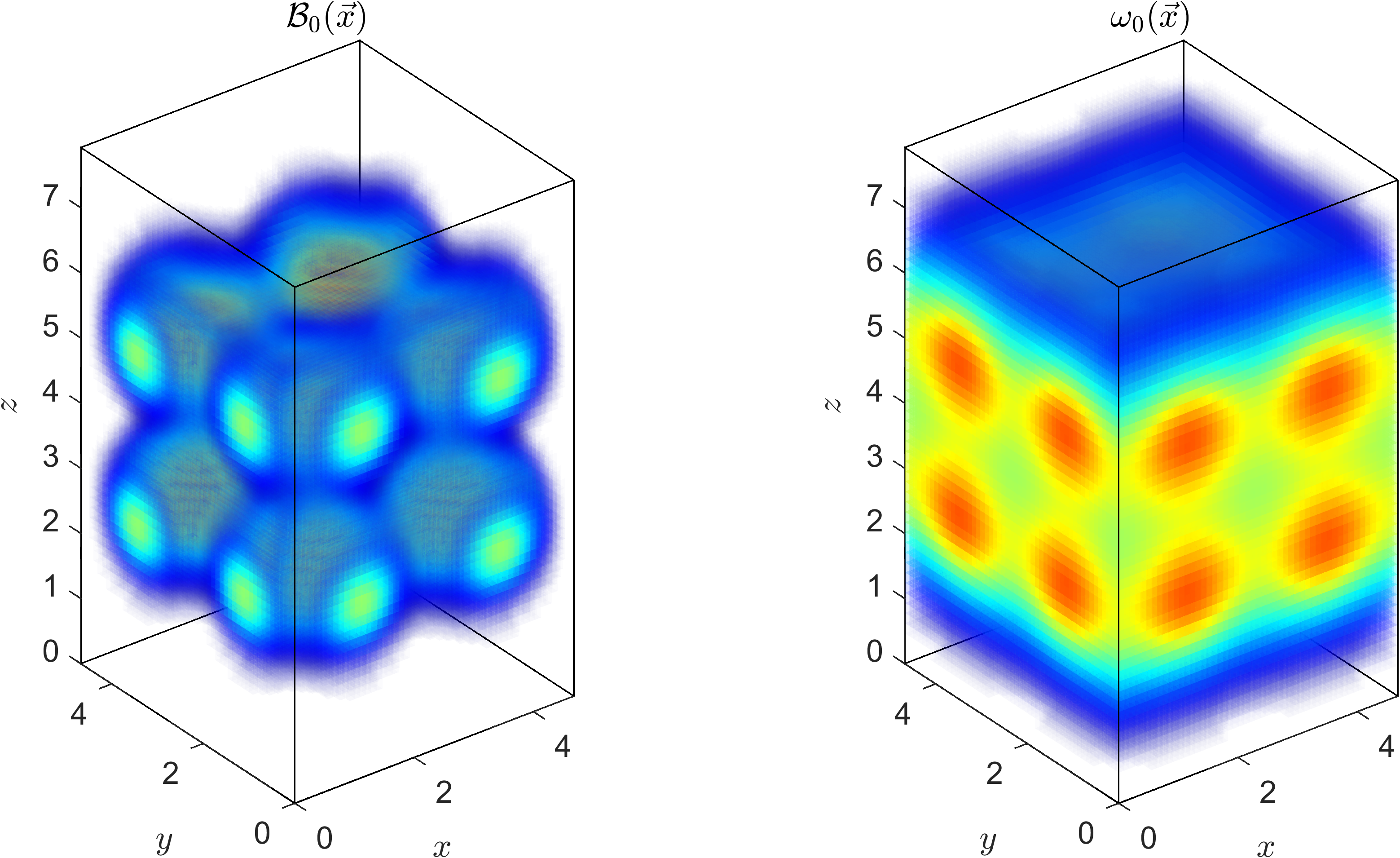}
	\caption{multiwall crystal}
	\label{fig: Multiwall}
	\end{subfigure}
	~
	\begin{subfigure}[b]{0.45\textwidth}
	\includegraphics[width=\textwidth]{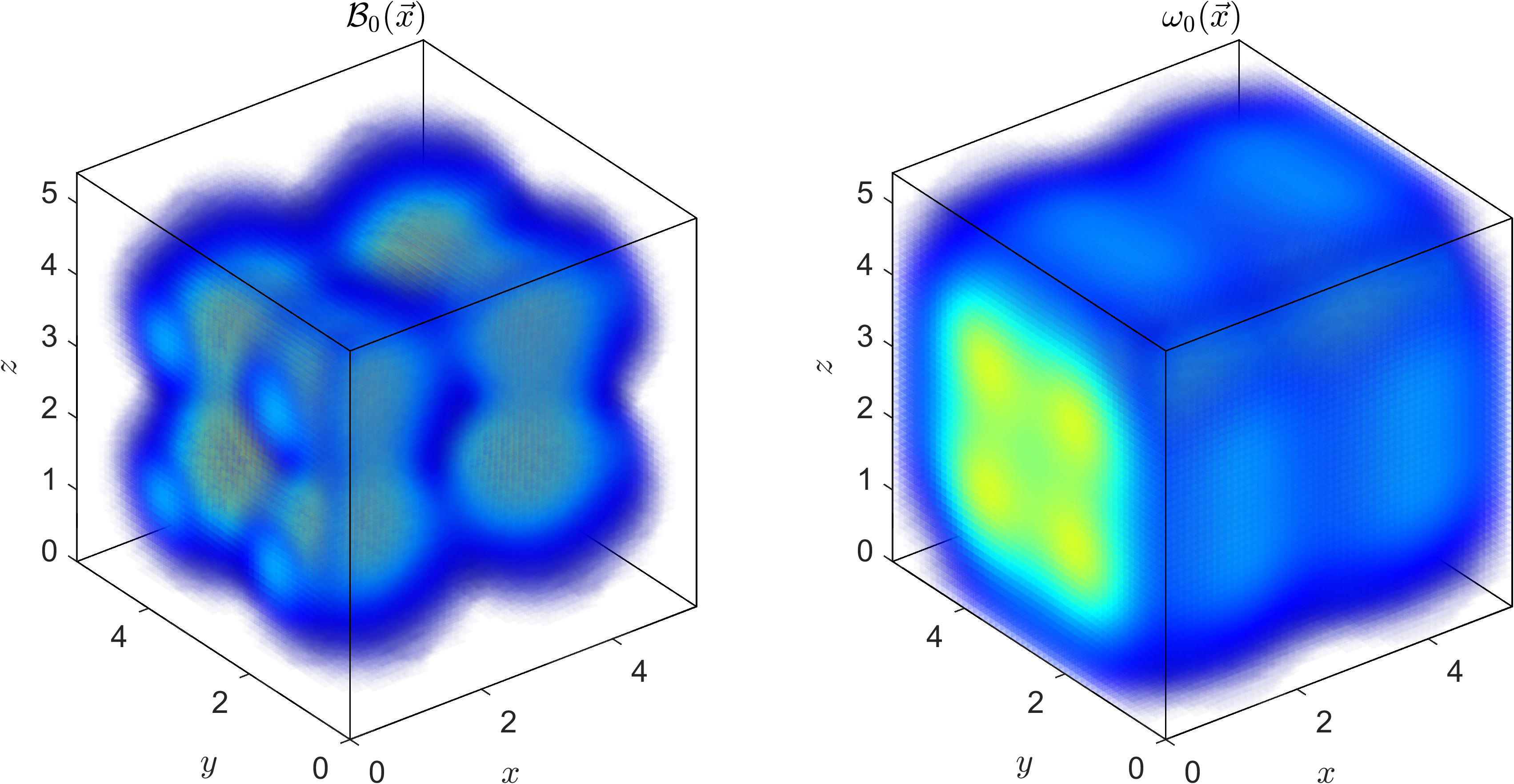}
	\caption{chain crystal}
	\label{fig: Chain}
	\end{subfigure}
	\caption{Baryon density $\mathcal{B}_0(\vec{x})$ and omega density $\omega_0(\vec{x})$ plots of the four crystalline solutions for the coupling constant $c_{\omega} = 14.34$.}
	\label{fig: Crystal densities}
\end{figure}

The previous sections have described our numerical algorithm that constructs skyrmion crystals by relaxing a choice of initial configuration.
We now present the crystals obtained using this algorithm.
As in \cite{Leask_2023}, our initial configurations are based on the $\textup{SC}_{1/2}$ crystal in the standard Skyrme model with no pion mass and no $\omega$-mesons.
We use the approximate solution $\varphi^\text{app.}$ of Castillejo \textit{et al.} \cite{Castillejo_1989},
\begin{equation}
    \varphi^\text{app.}_0  = - c_1 c_2 c_3, \quad \varphi^\text{app.}_1 = s_1 \sqrt{1-\frac{s_2^2}{2}-\frac{s_3^2}{2}+\frac{s_2^2 s_3^2}{3}},
\label{eq: 1/2 crystal Castillejo}
\end{equation}
with $s_i=\sin 2\pi x_i$, $c_i=\cos 2\pi x_i$, and $\varphi^\text{app.}_2,\varphi^\text{app.}_3$ obtained by cyclic permutation.
This defines a Skyrme field on $\mathbb{R}^3/\mathbb{Z}^3$ with $B=4$.
As in \cite{Leask_2023}, we generate a range of initial conditions $Q\varphi^\text{app.}$ using an $\SO(4)$ matrix $Q$.
The four specific choices that we make for $Q$ are:
\begin{equation}
\begin{aligned}
Q_{1/2}&=\mathrm{Id},&
Q_{\alpha}&=\frac{1}{\sqrt{3}}\begin{pmatrix}\begin{array}{cccc}0&1&1&1\end{array}\\ \ast \end{pmatrix}\\
Q_{\mathrm{multiwall}}&=\begin{pmatrix}\begin{array}{cccc}0&0&0&1\end{array}\\ \ast \end{pmatrix},&
Q_{\mathrm{chain}}&=\frac{1}{\sqrt{2}}\begin{pmatrix}\begin{array}{cccc}0&0&1&1\end{array}\\ \ast \end{pmatrix},
\end{aligned}
\label{eq: Chiral Q transformations}
\end{equation}
with the remaining rows (denoted by an asterisk) being determined by the Gram-Schmidt process.
These choices are motivated by the principle of symmetric criticality \cite{Leask_2023}.
The initial metric is given by $g_{ij}=L^2\delta_{ij}$ for suitably chosen $L$.
Following \cite{Gudnason_2020}, we set the initial configuration for the $\omega$-meson to be $\omega = -c_{\omega}\mathcal{B}_0$.
The resulting initial conditions are invariant under distinct subgroups of the symmetry group of the energy functional, so flow to distinct critical points.

The energy \eqref{eq: Static energy} and constraint \eqref{eq: Omega field equation} involve two dimensionless parameters: $c_\omega$ and $m$.
We used three different parameter choices that have been proposed in the literature \cite{Nappi_1984,Sutcliffe_2009,Gudnason_2020}.
Adkins and Nappi \cite{Nappi_1984} chose the value $c_\omega=98.4$ by fitting the masses of the nucleon and the delta resonance.
Sutcliffe \cite{Sutcliffe_2009} chose the value $c_\omega=34.7$ by fitting the pion decay constant $F_\pi$ and the mass of helium-4 to their experimental values.
Finally, Gudnason and Speight \cite{Gudnason_2020} chose the value $c_\omega=14.34$ motivated by a range of considerations.
In all calibrations, the parameter $m=m_\pi/m_\omega$ is close to its experimental value 0.176.
For more details, see Table \ref{tab: Omega crystal comparison}.

The results of our relaxation algorithm are given in Table \ref{tab: Omega crystal comparison}.
Plots of the baryon density and $\omega$ field are shown in Figure \ref{fig: Crystal densities} for $c_\omega=14.34$ (pictures for other calibrations are similar).
The 1/2 crystal always has a higher energy than the other three, but the $\alpha$, chain, and multiwall crystals are very close in energy and their relative ordering might depend on $c_\omega$.
For $c_\omega=14.34$ and $98.4$ the multiwall crystal appears to have the lowest energy.
For $c_\omega=34.7$ the chain crystal may have a lower energy, but the numerical values are too close to be confident of this.
For comparison, in the Skyrme model with no $\omega$ meson the multiwall-crystal has lowest energy \cite{Leask_2023}.

We have not explored how the crystal energies depend on the parameter $m_\pi/m_\omega$, but insight into this can be gained from the standard Skyrme model.
In the standard Skyrme model the crystal energies coalesce as $m_\pi$ tends to $0$, and when $m_\pi=0$ all four are related by $\textup{SO}(4)$ chiral rotations.
We expect similar behaviour in the omega-meson model.

As in \cite{Leask_2023}, the fundamental domain of the lattice $\Lambda$ is not cubic for the multiwall and chain crystals.
For the multiwall crystal the two equal side lengths are shorter than the third side, while for the chain crystal they are longer.

Finally, we note that the energies of the crystals are all greater than the bound \eqref{eq:bound} derived in Proposition \ref{prop:bound} by a factor of at least 3.5.
This is unsurprising, as the derivation of the bound ignores most terms in the energy.
The discrepancy seems to be greater for the lowest-energy solutions; this is because the bound \eqref{eq:bound} depends on the volume of the lattice fundamental domain, and solutions with lower energy happen to have large volumes.
We expect the bound \eqref{eq:bound} to be more effective when the size of the lattice fundamental domain is constrained to be small.

\begin{table*}[t]
    \centering
    \begin{tabular}{|c|c|c|c|c|c|c|c|}
        \hline
        Crystal & $c_{\omega}$ & $F_{\pi}\,(\textup{MeV})$ & $m_{\pi}\,(\textup{MeV})$ & $m_\omega\,(\textup{MeV})$ & $E$ & $E_0\,(\textup{MeV})$ & $n_0\,(\textup{fm}^{-3})$ \\
        \hline
        $\textup{SC}_{1/2}$ & $98.4$ & $124.0$ & $138.0$ & $782.0$ & $145.7761$ & $716.6$ & $0.128$ \\
        $\alpha$ & $98.4$ & $124.0$ & $138.0$ & $782.0$ & $145.4590$ & $715.0$ & $0.125$ \\
        chain & $98.4$ & $124.0$ & $138.0$ & $782.0$ & $145.4526$ & $715.0$ & $0.125$ \\
        multiwall & $98.4$ & $124.0$ & $138.0$ & $782.0$ & $145.4477$ & $715.0$ & $0.125$ \\
        \hline
        $\textup{SC}_{1/2}$ & $34.7$ & $186.0$ & $138.0$ & $782.0$ & $77.8067$ & $860.6$ & $0.526$ \\
        $\alpha$ & $34.7$ & $186.0$ & $138.0$ & $782.0$ & $77.7126$ & $859.6$ & $0.526$ \\
        multiwall & $34.7$ & $186.0$ & $138.0$ & $782.0$ & $77.6870$ & $859.3$ & $0.515$ \\
        chain & $34.7$ & $186.0$ & $138.0$ & $782.0$ & $77.6758$ & $859.1$ & $0.513$ \\
        \hline
        $\textup{SC}_{1/2}$ & $14.34$ & $139.8$ & $43.91$ & $249.5$ & $47.2632$ & $925.6$ & $0.060$ \\
        chain & $14.34$ & $139.8$ & $43.91$ & $249.5$ & $47.0900$ & $922.2$ & $0.052$ \\
        $\alpha$ & $14.34$ & $139.8$ & $43.91$ & $249.5$ & $47.0867$ & $922.1$ & $0.051$ \\
        multiwall & $14.34$ & $139.8$ & $43.91$ & $249.5$ & $46.8397$ & $917.3$ & $0.047$ \\
        \hline      
    \end{tabular}
    \caption{Comparison of the four crystalline solutions for the three different sets of parameters ($c_{\omega}=98.4$ \cite{Nappi_1984}, $c_{\omega}=34.7$ \cite{Sutcliffe_2009} and $c_{\omega}=14.34$ \cite{Gudnason_2020}).}
    \label{tab: Omega crystal comparison}
\end{table*}


\section{Bethe--Weizs\"acker semi-empirical mass formula}
\label{sec:SEMF}

In this section we use skyrmion crystals to estimate coefficients in the Bethe--Weizs\"acker semi-empirical mass formula.  This is an approximate formula for the binding energy of a nucleus with baryon number $B$ and takes the form
\begin{equation}
\label{eq: SEMF}
    E_b = a_V B - a_S B^{2/3} - a_C \frac{Z(Z-1)}{B^{1/3}} - a_A \frac{(N-Z)^2}{B} + \delta(N,Z).
\end{equation}
Here $Z$ is the number of protons and $N=B-Z$ the number of neutrons.
We will focus just on the first two terms, which are associated with the volume and surface area of the nucleus.
Typical empirically-determined values for their coefficients are $a_V=15.7-16.0\,\textup{MeV}$ and $a_S=17.3-18.4\,\textup{MeV}$ \cite{Reinhard_2006}.

We estimate these coefficients using the $\alpha$-crystal in the calibration of Gudnason--Speight \cite{Gudnason_2020}.
Following \cite{Baskerville_1996_2}, we model a $B=4N^3$ skyrmion as a cubic arrangement of $B=4$ skyrmions.
This can be regarded as a chunk of the $\alpha$-crystal.
To a first approximation, its energy is $E=BE_{\mathrm{crystal}}$, where $E_{\mathrm{crystal}}=922.1\,\text{MeV}$ is the energy per baryon number of the $\alpha$-crystal determined in the previous section.
To make a better approximation, we add on a term $6N^2E_{\mathrm{face}}$ representing the surface energy, in which $E_{\mathrm{face}}$ represents the surface energy of one face of one cubic $B=4$ skyrmion.
To calculate the binding energy, we subtract this from $B$ times the classical energy $E_1$ of a 1-skyrmion.
This leads to the formula,
\begin{equation}
\label{eq: skyrme SEMF}
E_b=(E_1-E_{\mathrm{crystal}})B - \frac{3}{\sqrt[3]{2}}E_{\mathrm{face}}B^{2/3},
\end{equation}
from which we can read off the coefficients $a_V$ and $a_S$.

We calculated $E_1$ using a fully three-dimensional arrested Newton flow and obtained the value $937.7\,\mathrm{MeV}$ in agreement with \cite{Gudnason_2020}.
Thus it remains to calculate $E_{\mathrm{face}}$.
We have done so using a method developed in \cite{Leask_2022}.
We regard the cubic lattice of $B=4$ skyrmions as vertical stack of horizontal layers, each layer being a square array of $B=4$ skyrmions.
A horizontal slab consisting of $n$ layers is doubly-periodic and has charge $4n$ in its fundamental domain.
The energy contained in a fundamental domain can be estimated as
\begin{equation}
\label{eq: slab energy}
4nE_{\mathrm{crystal}}+2E_{\mathrm{face}},
\end{equation}
because there are $n$ cubic $B=4$ skyrmions and two exposed faces.
On the other hand, for any given $n$, the energy can be calculated precisely using the same method as was used to calculate the energy of the Skyrme crystal, the only difference being that the domain is $\mathbb{R}\times\mathbb{T}^2$ rather than than $\mathbb{T}^3$, and only the components of the metric associated with $\mathbb{T}^2$ need to be varied.
We have constructed these doubly-periodic slabs for a range of values of $n$ using our relaxation algorithm.

By comparing these numerically-determined energies with the approximate formula \eqref{eq: slab energy} using a trust region reflective algorithm we have estimated the coefficient $E_{\mathrm{face}}$ to be $7.8\,\mathrm{MeV}$.
Then by comparing equations \eqref{eq: SEMF} and \eqref{eq: skyrme SEMF} we obtain the coefficients,
\begin{equation}
a_V=E_1-E_{\mathrm{crystal}}=15.6\,\mathrm{MeV},\quad a_S=\frac{3}{\sqrt[3]{2}}E_{\mathrm{face}}=18.6\,\mathrm{MeV}.\label{eq:SEMF predictions}
\end{equation}
The resulting energy per nucleon curve is plotted in Figure \ref{fig: SEMF from Skyrme crystals}.
For comparison, we have also plotted the energy per nucleon for the three cubic skyrmions with $B=4N^3$ and $N=1,2,3$ which have been calculated using arrested Newton flow, with initial configuration constructed using the rational map and product approximations.
These are all close to the fitted curve, confirming the validity of the approximate formula \eqref{eq: skyrme SEMF}.
We also remark that the $\alpha$-particle clustering seen here matches the clustering structure of light nuclei, and was also observed in the Skyrme model coupled to $\rho$-mesons \cite{Naya_2018}.

\begin{figure*}[t]
    \centering
    \includegraphics[width=\textwidth]{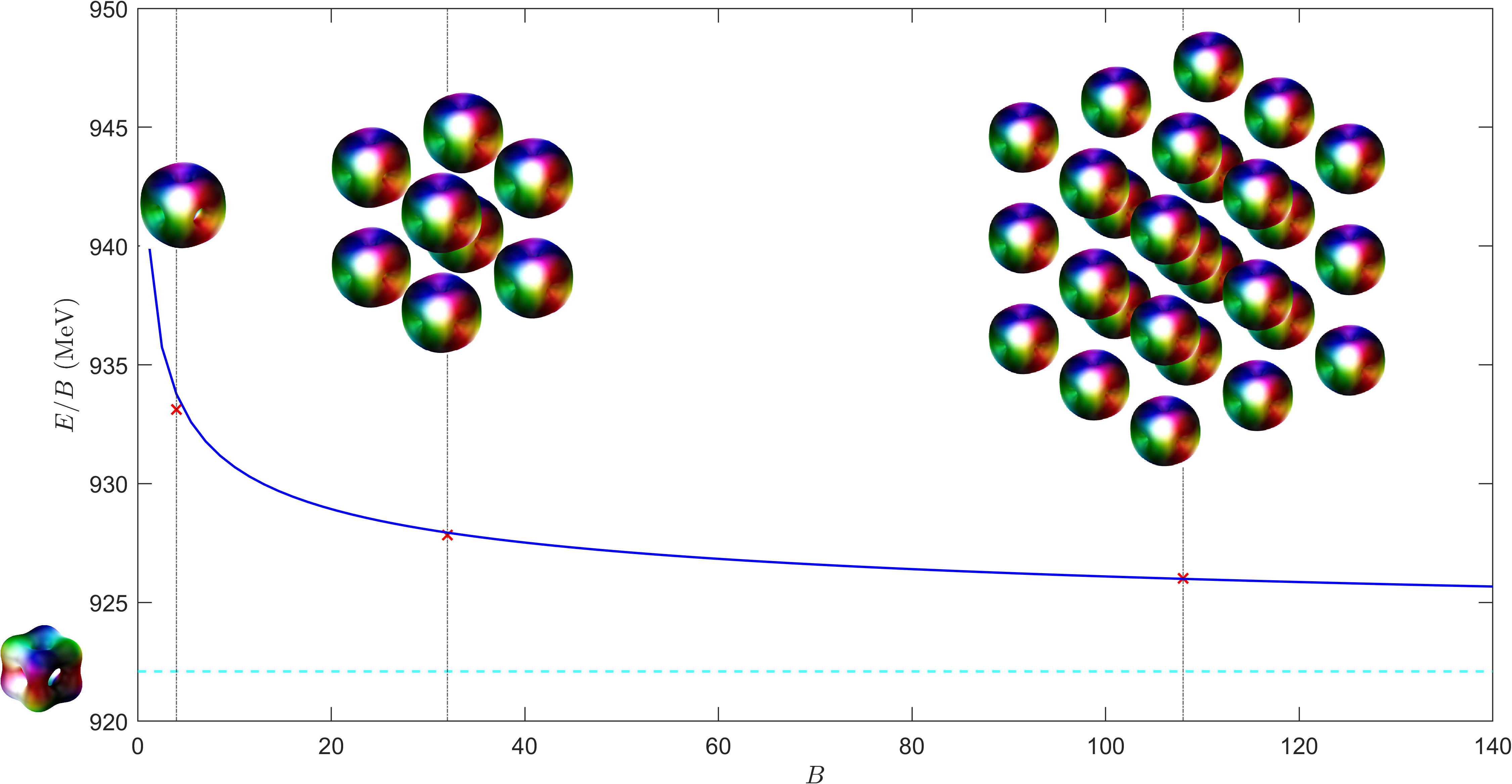}
    \caption{Plot of the Bethe--Weizs\"acker SEMF from the $\alpha$-particle approximation for the $\omega$-Skyrme model.}
    \label{fig: SEMF from Skyrme crystals}
\end{figure*}

Our predicted values \eqref{eq:SEMF predictions} are in close agreement with the empirically-determined values.
They are also a substantial improvement on the values $a_V=136\,\mathrm{MeV}$, $a_S=320\,\mathrm{MeV}$ obtained in the standard Skyrme model with massless pions \cite{Baskerville_1996_2}.
This is consistent with the observation made in \cite{Gudnason_2020} that including the $\omega$-meson substantially improves predictions of classical binding energies.

The remaining coefficients in the semi-empirical mass formula \eqref{eq: SEMF} have been investigated using Skyrme models elsewhere.
Ma \textit{et al.}\ \cite{Ma_2019} investigated the Coulomb energy in the standard Skyrme model.
They found that $a_C=0.608\,\textup{MeV}$, which is in excellent agreement with the experimentally determined value of $a_C=0.625\,\textup{MeV}$.
The asymmetry coefficient $a_A$ was calculated using the sextic Skyrme model in \cite{Leask_2024}.
By relating this to the symmetry energy of nuclear matter the value $a_A=23.8\,\textup{MeV}$ was obtained, which agrees extremely well with the experimental value $a_A=23.7\,\textup{MeV}$.
It would be interesting to calculate these coefficients also in the omega-meson model, but doing so is beyond the scope of this paper.


\section{Incompressibility of Nuclear Matter}
\label{sec:incompressibility}

We have seen that the Skyrme crystal in the $\omega$-meson model provides a reasonable model of binding energies of finite nuclei.
In this section we turn our attention to properties of infinite nuclear matter.

Consider isospin symmetric nuclear matter at zero temperature, and let $n_B$ be the baryon density, i.e.\ the number of protons and neutrons per unit volume.  The (symmetric) energy per baryon of such matter can be approximated about the nuclear saturation density $n_0$ by use of a power series expansion \cite{Garg_2018}:
\begin{equation}
E(n_B)/B = E_0 + \frac{1}{2} K_0 \frac{(n_B-n_0)^2}{9n_0^2} + \mathcal{O} \left( (n_B-n_0)^3 \right),
\label{eq: Symmetric energy expansion}
\end{equation}
where the first term, associated to the nuclear saturation point $n_0$, is identified with the saturation energy $E_0 = E(n_0)/B$.
The nuclear saturation density $n_0$ is defined to be the nuclear density such that $(\partial E)/(\partial n_B)|_{n_B=n_0}=0$.
There is no linear term since symmetric nuclear matter reaches a minimum of the energy at saturation.
The next term is the one of interest, it is the nuclear incompressibility coefficient, or compression modulus, $K_0$, which can be obtained from the expansion \eqref{eq: Symmetric energy expansion},
\begin{equation}
    \left. K_0 = \frac{9 n_0^2}{B} \frac{\partial^2 E}{\partial n_B^2} \right|_{n_0}.
\end{equation}
This is a fundamental quantity in nuclear physics as it is a measure of nuclear resistance under pressure at the saturation point, and imposes significant constraints on the nuclear matter equation of state.

To compute the compression modulus in the Skyrme model, we need to construct skyrmion crystals with a range of baryon densities.
In practice, we do this by minimizing the energy of the crystal with the volume of the fundamental cell of the lattice constrained to a constant value $V$.
The baryon density of the resulting solution is $n_B=B/V$, with $B$ being the baryon number per unit cell.
The algorithm that we used to solve the constrained energy-minimization problem is similar to the algorithm used for the unconstrained problem, except that the stress-energy tensor $S_{ij}$ is replaced by its projection, $S_{ij}-\frac{1}{3}S_{kl}g^{kl}g_{ij}$.  
This ensures that the volume of the fundamental domain, proportional to $\sqrt{\det g}$, is unchanged by arrested Newton flow.

We computed the compression modulus for the multiwall crystal, as this is the crystal with the lowest energy and hence the best candidate to model nuclear matter.
As in the previous section, we used the calibration of Gudnason--Speight \cite{Gudnason_2020}.  

The resulting data are plotted in Fig.~\ref{fig: Compression modulus} and we determine a compression modulus value of $K_0=370\,\textup{MeV}$, roughly 1/3 of the value $E_0=917.3\,\text{MeV}$ for the saturation energy obtained from our model.
This result is comparable in magnitude with other theoretical calculations of the compression modulus, and also with the range $250 < K_0 < 315 \, \textup{MeV}$ determined by a recent survey of experimental data \cite{Stone_2014}.

\begin{figure*}[t]
\centering
\includegraphics[width=0.6\textwidth]{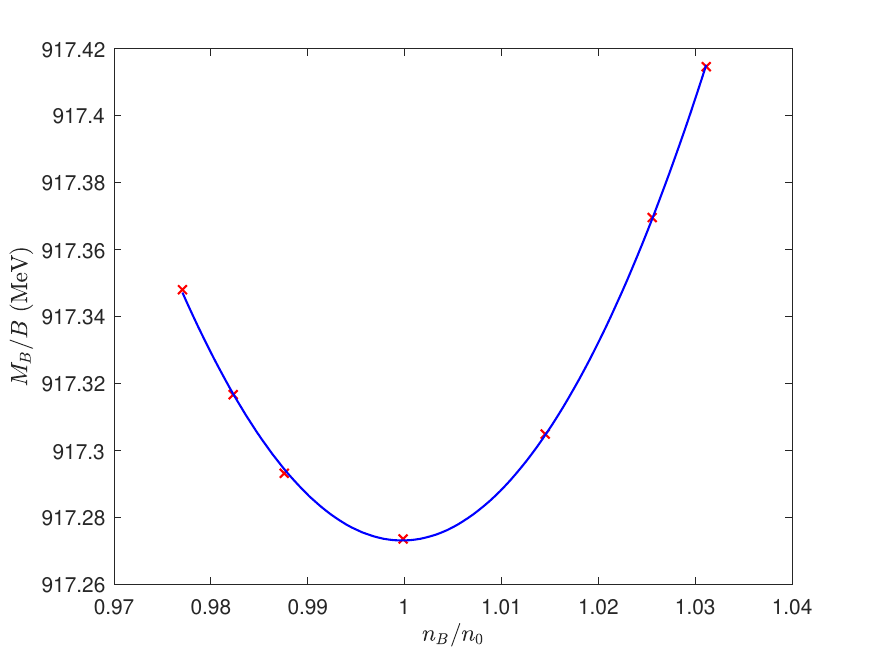}
\caption{The energy per baryon $E/B$ of the multi-wall crystal for various baryon densities $n_B$ near saturation $n_0$.}
\label{fig: Compression modulus}
\end{figure*}

This result is a significant improvement on values obtained in other Skyrme models.
For example, a recent calculation based on the sextic Skyrme model (and also using the multiwall crystal) obtained a value $K_0=1169\,\text{MeV}$ and a ratio $K_0/E_0=1.28$.
We have also carried out a calculation in the standard Skyrme model with massless pions, and obtained a value $K_0/E_0=0.973$.
Both values are much higher than the value $K_0/E_0\approx 1/4$ obtained from experiment.

The fact that our value for $K_0$ is small in comparison with other Skyrme models indicates that the minimum of $E(n_B)/B$ at $n_0=0.047\,\textup{fm}^{-3}$ is shallow, so energies do not change much as density is varied.
This is reminiscent of the fact that the calibration of \cite{Gudnason_2020} produces low classical binding energies.
It is also consistent with the observation that in other versions of the Skyrme model the multiwall crystal exhibits the smallest variation in energy amongst all crystals as the density is decreased \cite{Leask_2023}.
So our favourable result for the compression modulus can be attributed both to the success of the $\omega$-meson model in producing low binding energies and to particular properties of the multiwall crystal.


\section{Concluding remarks}
\label{sec:conclusion}

In this paper we have constructed crystalline configurations in the $\omega$-meson variant of the Skyrme model and investigated their applications to cold dense nuclear matter.  Our construction was based on a new algorithm that combined methods developed in \cite{Gudnason_2020} and \cite{Leask_2023}.  It minimizes energy with respect to variations in both the Skyrme field and the period lattice.

Using these new crystals, we have calculated coefficients in the Bethe--Weizs\"acker semi-empirical mass formula and the nuclear matter incompressibility coefficient.  In both cases we obtained results that are comparable with other theoretical models and with experimental evidence.  This is a substantial improvement on previous studies based on other variants of the Skyrme model.

The $\omega$-meson variant of the Skyrme model has also been successful in reproducing binding energies of light nuclei \cite{Gudnason_2020}.  So this variant of the Skyrme model shows promise as a model of nuclear physics, and is worthy of further study.


\section*{Acknowledgments}

P.~Leask is supported by a Ph.D. studentship from UKRI, Grant No. EP/V520081/1.
We would like to thank the organisers of the Solitons (non)Integrability Geometry XI \href{https://sig.fais.uj.edu.pl/main}{(SIG XI)} conference, where this paper was conceptualised.




\bibliographystyle{JHEP.bst}
\bibliography{bib.bib}

\end{document}